\documentclass[pre,
twocolumn,
groupedaddress]{revtex4-1}
\usepackage{graphicx}
\usepackage{amsmath, amssymb, amsthm}
\usepackage{epsfig, color}
\usepackage{systeme}
\usepackage{multirow, comment, enumerate}
\usepackage{hyperref}
\usepackage{enumitem}
\usepackage{esint}
\usepackage{lipsum}
\usepackage{bm}

\definecolor{myblue}{RGB}{0,50,200}
\hypersetup{
	colorlinks,
	citecolor=myblue,
	linkcolor=myblue,
	urlcolor=myblue
}
\allowdisplaybreaks

\theoremstyle{definition}
\newtheorem{theorem}{\textit{Theorem}}
\theoremstyle{definition}

\newcommand{\mca}{\mathcal}
\newcommand{\mbb}{\mathbb}
\newcommand{\mrm}{\mathrm}

\begin{document}
\title{Diffusion-dynamics laws in stochastic reaction networks}
\author{Tan Van Vu}
\email{tan@biom.t.u-tokyo.ac.jp}
\affiliation{Department of Information and Communication Engineering, Graduate School of Information Science and Technology, The University of Tokyo, Tokyo 113-8656, Japan}

\author{Yoshihiko Hasegawa}
\email{hasegawa@biom.t.u-tokyo.ac.jp}
\affiliation{Department of Information and Communication Engineering, Graduate School of Information Science and Technology, The University of Tokyo, Tokyo 113-8656, Japan}

\date{\today}

\begin{abstract}
Many biological activities are induced by cellular chemical reactions of diffusing reactants. 
The dynamics of such systems can be captured by stochastic reaction networks. 
A recent numerical study has shown that diffusion can significantly enhance the fluctuations in gene regulatory networks. 
However, the universal relation between diffusion and stochastic system dynamics remains veiled. 
Within the approximation of reaction-diffusion master equation (RDME), we find general relation that the steady-state distribution in complex balanced networks is diffusion-independent.
Here, complex balance is the nonequilibrium generalization of detailed balance. 
We also find that for a diffusion-included network with a Poisson-like steady-state distribution, the diffusion can be ignored at steady state. 
We then derive a necessary and sufficient condition for networks holding such steady-state distributions. 
Moreover, we show that for linear reaction networks the RDME reduces to the chemical master equation, which implies that the stochastic dynamics of networks is unaffected by diffusion at any arbitrary time.
Our findings shed light on the fundamental question of when diffusion can be neglected, or (if nonnegligible) its effects on the stochastic dynamics of the reaction network.
\end{abstract}

\maketitle

\section{Introduction}
Diverse biological phenomena, such as cellular signal transductions and gene expression systems, are commonly studied by stochastic reaction network modeling \cite{Blake.2003,Lan.2007,Tsimring.2014}. 
These systems involve a set of reactant species which react through several channels. 
In most of the existing studies, such systems are often assumed to be well mixed, meaning that the diffusion coefficients of the reactants are infinitely large \cite{Johan.2003,Carlos.2007,Tostevin.2010,Kyung.2013,Pilkiewicz.2016,Samanta.2017,Ouldridge.2017}. 
However, experiments have shown that reactants in cells diffuse at considerably low rates \cite{Ellis.2001}, and that the smallest timescale of the system is a little larger than the timescale of molecular diffusion. 
In such cases, the well-mixed assumption cannot accurately obtain the stochastic dynamics of the system. 
For example, living cells continuously receive signals at their receptors, which are subsequently transmitted to the nucleus through biochemical reaction networks \cite{Kholodenko.2006,Cheong.2011,Becker.2015,Hasegawa.2018}. 
This process is strongly influenced by extrinsic and intrinsic noise arising from fluctuations in the input and reactions. 
These effects induce unavoidable fluctuations in the biomolecule concentrations, which deteriorate the fidelity of information transfer \cite{Elowitz.2002,Ozbudak.2002}. 
By accurately evaluating the fluctuations, we would better understand the mechanism underlying signal transmission in cells. 
In a numerical study of gene regulatory networks, Ref. \cite{Jeroen.2006} showed that the fluctuations are larger in the model with diffusion than in its well-mixed counterpart.
Thus, how diffusion relates to the stochastic dynamics of reaction networks is a pertinent question.
Recently, Ref. \cite{Smith.2018} has numerically studied the effects of diffusion on single-cell variability in multicellular organisms, and the limits of slow and fast diffusion have been investigated.

Two commonly used models for studying stochastic reaction-diffusion systems are the reaction-diffusion master equation (RDME) \cite{Gardiner.2009} and the Smoluchowski model \cite{Von.1917}.
The RDME, which is a mesoscopic model, is an extension of the nonspatial chemical master equation (CME) \cite{Gardiner.2009} and can be interpreted as an asymptotic approximation to spatially continuous stochastic reaction-diffusion models \cite{Samuel.2009.PRE}.
The RDME has been successfully applied in studying many biological systems \cite{Howard.2003,Fange.2006,Lawson.2013,Smith.2018}.
It is worth noting that the Langevin equation, which can be derived from an equivalent Fokker--Planck equation or the Poisson representation, can handle continuum-limit diffusion in reaction networks \cite{Gardiner.2009,Benitez.2016}.
However, the Langevin equation is applicable to biochemical reactions occurring in infinite space with no physical boundary, which is unrealistic in biological cells.

In the present work, we investigate the relations between diffusion and the stochastic dynamics of reaction networks within a physical reflecting boundary.
In this system, reactants diffuse within a closed 3-dimensional space without escaping. 
With the aid of the RDME, we find an intriguing law stating that diffusion does not affect the steady-state distribution of complex balanced networks, which have a Poisson-like distribution. 
Our proof reveals that if the network presents a steady-state distribution of product-of-Poissonians form, diffusion can be neglected. 
We then calculate the necessary and sufficient conditions for such steady-state network distributions. 
We also find another result, wherein steady state is not a requirement.
Specifically, we prove that for linear reaction networks, one can derive the CME from the RDME, which indicates that diffusion can be ignored in this case.
This result can be restated as follows: the stochastic dynamics of linear networks are diffusion-independent, which is consistent with the Smoluchowski model.
In addition, we perform stochastic simulations on both linear and nonlinear networks to verify our results.

\section{Models}
We consider a general reaction network consisting of $N$ reactant species $X_1,\dots,X_N$ and $K$ reactions $R_1,\dots,R_K$. Assume that all reactions occur inside a cell with fixed volume $\Omega$, and that reaction $R_j~(1\leq j\leq K)$ is of the form
\begin{equation}\label{eq:react.form}
s_{1j}X_1+\dots+s_{Nj}X_N\xrightarrow{k_j}r_{1j}X_1+\dots+r_{Nj}X_N,
\end{equation}
where $s_{ij}, r_{ij}\in\mbb{N}_{\geq 0}$ are the stoichiometric coefficients and $k_j\in\mbb{R}_{>0}$ is the macroscopic reaction rate. Here, $\mbb{N}_{\geq 0}$ denotes the set of nonnegative integers. $\mbb{R}_{>0}$ and $\mbb{R}_{\geq 0}$ are defined analogously. If $\sum_{i=1}^{N}s_{ij}\leq 1$ for all $j=1,\dots,K$, then the reaction network is linear; otherwise, it is nonlinear. The state of the system is fully determined by the molecule-number vector of all reactant species in the system, $\bm{n}=[n_1,\dots, n_N]^{\top}$, where $n_i\in\mbb{N}_{\geq 0}$ is the number of molecules of species $X_i$. Assuming mass-action kinetics, the time evolution of a well-mixed system can be described by the following chemical master equation (CME):
\begin{equation}\label{eq:cme}
\partial_t P(\bm{n},t)=\sum_{j=1}^{K}(\bm{\mbb{E}}^{-\bm{V}_j}-1)f_{j}(\bm{n},\Omega)P(\bm{n},t),
\end{equation}
where $\bm{V}=[r_{ij}-s_{ij}]\in\mbb{Z}^{N\times K}$ is a stoichiometric matrix, $\bm{V}_{j}$ denotes the $j^{\mathrm{th}}$ column of matrix $\bm{V}$, and $\bm{\mbb{E}}^{\bm{x}}$ is an operator that replaces $\bm{n}$ with $\bm{n}+\bm{x}$. $P(\bm{n},t)$ is the probability of the system being in state $\bm{n}$ at time $t$, and the propensity function $f_j(\bm{n},\Omega)$ of reaction $R_j$ is given by
\begin{equation}\label{eq:propensity.func}
f_j(\bm{n},\Omega)=k_j\Omega^{1-\sum_{i=1}^{N}s_{ij}}\prod_{i=1}^{N}\frac{n_i!}{(n_i-s_{ij})!}.
\end{equation}
To include diffusion in stochastic spatial dynamics, many researchers apply the RDME, in which space is partitioned discretely into many voxels. 
It is known that the RDME is accurate if an appropriate combination of the time- and length-scale is chosen \cite{Elf.2004,Dobrzynski.2007,Samuel.2009.SIAM,Samuel.2009.PRE,Radek.2009}.
We assume from now on that the volume of the system is optimally divided into small voxels and as such, the RDME yields a good description of the time evolution of the probability distribution.
Diffusion then occurs among the voxels, and the reaction can occur within the same voxel considered to be a well-mixed system. 
Assume that the volume $\Omega$ is divided into a set $\mca{V}$ of voxels labeled by integers $v=1,2,\dots,|\mca{V}|$. 
Each voxel $v$ occupies a constant volume $\omega$ and contains $n_{vi}$ molecules of reactant species $X_i$. The state vector of voxel $v$ is denoted as $\bm{n}_{v}=[n_{v1},\dots,n_{vN}]^{\top}$. 
The state of the whole system is then described as the molecule-number vector $\bm{n}$ of each species in each voxel, namely, $\bm{n}=[\bm{n}_{1}^{\top},\dots,\bm{n}_{|\mca{V}|}^{\top}]^{\top}$. 
We also define a vector $\bm{1}_{vi}\in\mbb{Z}^{|\mca{V}|N}$, in which the number of molecules of all species in all voxels is zero except for species $X_i$ in voxel $v$ (which is one), and a vector $\widetilde{\bm{V}}_{vj}\in\mbb{Z}^{|\mca{V}|N}$, in which all elements are zero except in voxel $v$ (which holds $\bm{V}_j$). 
As the diffusion of each species into neighboring voxels can be modeled as a first-order reaction, the diffusion-included reaction network can be described in the following form:
\begin{equation}
\begin{aligned}
&s_{1j}X_1^v+\dots+s_{Nj}X_N^v\xrightarrow{k_j}r_{1j}X_1^v+\dots+r_{Nj}X_N^v,\\
&X_{i}^{v}\xrightarrow{d_i}X_{i}^{v'},~\forall~1\leq i\leq N,v\in\mca{V},v'\in N_e(v),
\end{aligned}
\end{equation}
where $X_{i}^{v}$ refers to species $X_i$ in voxel $v$, $d_i$ is the diffusion rate of species $X_i$, and $N_e(v)$ is the set of voxels neighboring $v$. The stochastic dynamics of the system can then be described by the following RDME:
\begin{equation}\label{eq:rdme}
\begin{aligned}
\partial_tP(\bm{n},t)&=\sum_{v\in \mca{V}}\sum_{v'\in  N_e(v)}\sum_{i=1}^{N}\left(\bm{\mbb{E}}^{\bm{1}_{vi}-\bm{1}_{v'i}}-1\right)d_{i}n_{vi}P(\bm{n},t)\\
&+\sum_{v\in \mca{V}}\sum_{j=1}^{K}\left(\bm{\mbb{E}}^{-\widetilde{\bm{V}}_{vj}}-1\right)f_{vj}(\bm{n},\omega)P(\bm{n},t),
\end{aligned}
\end{equation}
where $f_{vj}(\bm{n},\omega)$ is the propensity function, given by
\begin{equation}\label{eq:spatial.propensity.func}
f_{vj}(\bm{n},\omega)=k_j\omega^{1-\sum_{i=1}^{N}s_{ij}}\prod_{i=1}^{N}\frac{n_{vi}!}{(n_{vi}-s_{ij})!}.
\end{equation}
In the large-diffusion limit, the RDME converges to the CME \cite{Stephen.2016}.

Before stating our results, we describe several existing concepts and results of deterministic reaction networks. For each reaction $R_j~(1\leq j\leq K)$, the linear combinations $\sum_{i=1}^{N}s_{ij}X_i$ and $\sum_{i=1}^{N}r_{ij}X_i$ of the species in Eq.~\eqref{eq:react.form} are called the complexes of the reaction. Defining $\mca{C}=\{C_1,C_2,\dots,C_M\}$ as the set of complexes, with $M=|\mca{C}|$, each reaction can be expressed as $C_i\xrightarrow{a_{ii'}}C_{i'}$, where $a_{ii'}$ denotes the reaction rate. For each $1\leq i,i'\leq M$, $a_{ii'}=0$ if $C_i\to C_{i'}$ is not present in the reaction network; otherwise, $a_{ii'}=k_j$ for some $j~(1\leq j\leq K)$. The matrix $\bm{A}\in\mbb{R}^{M\times M}$, called the \emph{Kirchhoff matrix} of the reaction network, is defined as follows:
\begin{equation}
[\bm{A}]_{ii'}=\begin{cases}
-\sum_{j=1}^{M}a_{ij}, & \mathrm{if}~i=i'\\
a_{i'i}, & \mathrm{if}~i\neq i'
\end{cases}.
\end{equation}
Let $\mca{X}=\{X_1,\dots,X_N\}$ be the set of species and $\mca{R}=\bigcup_{i,i':a_{ii'}>0}\{C_i\to C_{i'}\}$ be the set of reactions in the network. 
The triple $\{\mca{X},\mca{C},\mca{R}\}$ then defines a reaction network. 
A reaction network $\{\mca{X},\mca{C},\mca{R}\}$ is called \emph{weakly reversible} if for any reaction $C_{i}\to C_{i'}\in\mca{R}$, there exists a sequence of complexes $C_{i_1},\dots,C_{i_p}\in\mca{C}$ such that $C_{i'}\to C_{i_1},C_{i_1}\to C_{i_2},\dots,C_{i_p}\to C_{i}\in\mca{R}$.
One can construct a directed graph $\mca{G}$ corresponding to a reaction network in the following manner.
For each $1\leq i,i'\leq M$, draw a directed edge from $C_i$ to $C_{i'}$ if and only if $C_{i}\to C_{i'}\in\mca{R}$. 
We denote by $\ell$ the number of connected components of the underlying undirected graph of $\mca{G}$. 
The \emph{deficiency} of a reaction network is an integer defined as $\delta=|\mca{C}|-\ell-\mathrm{rank}(\bm{V})$.
According to Ref.~\cite{Feinberg.1995}, $\delta$ is always nonnegative.

In a deterministic system, the vector of species concentrations, $\bm{c}=[c_1,c_2,\dots,c_N]^{\top}\in\mbb{R}_{\geq 0}^{N}$, temporally evolves as described by the following differential equations, which express the different form of rate equations:
\begin{equation}
\partial_t\bm{c}=\bm{Y}\cdot\bm{A}\cdot\Psi(\bm{c}).
\end{equation}
Here, $\bm{Y}=[y_{ij}]\in\mbb{N}_{\geq 0}^{N\times M}$ is the matrix of stoichiometric compositions of the complexes, i.e., $y_{ij}$ is the stoichiometric coefficient of $C_{j}$ corresponding to species $X_{i}$, and $\Psi:\mbb{R}^{N}\mapsto\mbb{R}^{M}$ is a mapping given by
\begin{equation}
\Psi_{j}(\bm{c})=\prod_{i=1}^{N}c_i^{y_{ij}},~j=1,\dots,M.
\end{equation}
A reaction network is called \emph{complex balanced} at $\bm{c}\in\mbb{R}_{>0}^{N}$ if $\bm{A}\cdot\Psi(\bm{c})=0$. This condition means that for each complex $C_i\in\mca{C}$, $\sum_{C_i\to C_{i'}}a_{ii'}\Psi_{i}(\bm{c})=\sum_{C_{i'}\to C_i}a_{i'i}\Psi_{i'}(\bm{c})$. In this case, $\bm{c}$ is a positive equilibrium of the network.

\section{Results}
The following states our first result.
\begin{theorem}\label{theo:diffusion.independent}
If a reaction network is complex balanced, its steady-state distribution is unaffected by diffusion.
\end{theorem}
\begin{proof}
As the network is complex balanced, there exists a positive equilibrium $\bm{c}=[c_1,\dots,c_N]^{\top}\in\mbb{R}_{>0}^{N}$ such that $\bm{A}\cdot\Psi(\bm{c})=0$. 
We note that the only requirement in our proof is the existence of some $\bm{c}$ such that $\bm{A}\cdot\Psi(\bm{c})=0$. 
Here, $\bm{c}$ is not the steady-state concentration in the presence of diffusion.
Let $\Gamma\subseteq\mbb{N}_{\geq 0}^{N}$ be the state space of the network, which may depend on the initialization. 
First, we prove the following ansatz: that the steady-state distribution of the RDME is given by a product $P_{\Gamma}(\bm{n},t)$ of Poisson distributions:
\begin{equation}\label{eq:solution.prob.form}
P_{\Gamma}(\bm{n},t)=\begin{cases}
\mca{N}_{\Gamma}\prod_{v\in\mca{V}}\prod_{i=1}^{N}\frac{(\omega c_i)^{n_{vi}}}{n_{vi}!}, & \mathrm{if}~\sum_{v\in\mca{V}}\bm{n}_{v}\in\Gamma\\
0, & \mathrm{if}~\sum_{v\in\mca{V}}\bm{n}_{v}\notin\Gamma
\end{cases},
\end{equation}
where $\mca{N}_{\Gamma}$ is the normalizing constant. For each $\bm{n}_v\in\mbb{N}_{\geq 0}^{N}$, we define $P_{\Gamma}^*(\bm{n}_v,t)=\prod_{i=1}^{N}(\omega c_i)^{n_{vi}}/n_{vi}!$. $P_{\Gamma}(\bm{n},t)$ can then be expressed as $P_{\Gamma}(\bm{n},t)=\mca{N}_{\Gamma}\prod_{v\in\mca{V}}P_{\Gamma}^*(\bm{n}_v,t)$. Now, we need to show that $\partial_tP_{\Gamma}(\bm{n},t)=0$. Substituting $P_{\Gamma}(\bm{n},t)$ in Eq.~\eqref{eq:solution.prob.form} into Eq.~\eqref{eq:rdme}, the first term of the right-hand side becomes
\begin{equation}
\sum_{v\in \mca{V}}\sum_{v'\in  N_e(v)}\sum_{i=1}^{N}\left(\bm{\mrm{E}}^{\bm{1}_{vi}-\bm{1}_{v'i}}-1\right)d_{i}n_{vi}P_{\Gamma}(\bm{n},t)=0.\label{eq:diffusion.term}
\end{equation}
The second term on the right-hand side becomes the sum of the following values over all voxels $v\in\mca{V}$:
\begin{align*}
&\sum_{j=1}^{K}\left(\bm{\mbb{E}}^{-\widetilde{\bm{V}}_{vj}}-1\right)f_{vj}(\bm{n},\omega)P_{\Gamma}(\bm{n},t)\\
&=\mca{N}_{\Gamma}\prod_{v'\neq v}P_{\Gamma}^*(\bm{n}_{v'},t)\sum_{j=1}^{K}\left(\bm{\mbb{E}}^{-\bm{V}_{j}}-1\right)f_{j}(\bm{n}_v,\omega)P_{\Gamma}^*(\bm{n}_v,t).
\end{align*}
Exploiting the condition $\bm{A}\cdot\Psi(\bm{c})=0$, one can prove that \cite{Anderson.2010}
\begin{equation}
\sum_{j=1}^{K}\left(\bm{\mbb{E}}^{-\bm{V}_{j}}-1\right)f_{j}(\bm{n}_v,\omega)P_{\Gamma}^*(\bm{n}_v,t)=0.\label{eq:reaction.term}
\end{equation}
Therefore, the second term also disappears and we obtain the desired result $\partial_tP_{\Gamma}(\bm{n},t)=0$. Let $\widehat{\bm{n}}=\sum_{v\in \mca{V}}\bm{n}_{v}$ represent the number of molecules of all species, i.e., $\widehat{n}_i$ is the total number of molecules of species $X_i$ in the system. To complete our theorem, we compute the steady-state distribution $P_{\Gamma}(\widehat{\bm{n}})$, and show its diffusion-independence. For $\widehat{\bm{n}}\notin\Gamma$, obviously $P_{\Gamma}(\widehat{\bm{n}})=0$. For $\widehat{\bm{n}}\in\Gamma$, the explicit form of $P_{\Gamma}(\widehat{\bm{n}})$ is obtained as follows:
\begin{equation}
P_{\Gamma}(\widehat{\bm{n}})=\sum_{\bm{n}:\sum_{v}\bm{n}_v=\widehat{\bm{n}}}P_{\Gamma}(\bm{n},t)=\mca{N}_{\Gamma}\prod_{i=1}^{N}\frac{(\Omega c_i)^{\widehat{n}_i}}{\widehat{n}_i!}.\label{eq:explicit.distribution}
\end{equation}
As $\mca{N}_{\Gamma}=\left(\sum_{\widehat{\bm{n}}\in\Gamma}\prod_{i=1}^{N}\frac{(\Omega c_i)^{\widehat{n}_i}}{\widehat{n}_i!}\right)^{-1}$ does not depend on diffusion, the distribution $P_{\Gamma}(\widehat{\bm{n}})$ is also independent of diffusion.
The details of these derivations can be seen in Appendix \ref{app.theorem.1}.
\end{proof}
\begin{figure}[t]
	\centering
	\includegraphics[width=0.47\textwidth]{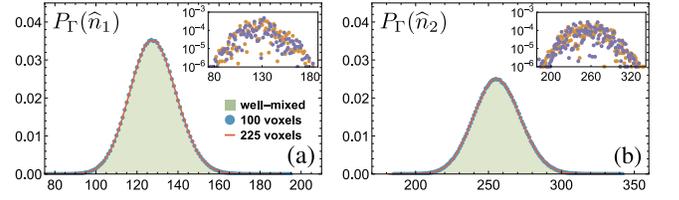}
	\caption{Steady-state distributions (a) $P_{\Gamma}(\widehat{n}_1)$ of species $X_1$ and (b) $P_{\Gamma}(\widehat{n}_2)$ of species $X_2$ of nonlinear reaction network. Each panel shows the distributions of the 1-voxel system (green region), 100-voxel system (blue dots), and 225-voxel system (red line). The parameters are $k_1=4,k_2=1,k_3=2,\Omega=128$. The diffusion rates of species $X_1,X_2$ are $d_{1}=1,d_{2}=2$ (100 voxels), and $d_{1}=2,d_{2}=1$ (225 voxels). Insets show the absolute probability differences $|P_{100}-P_1|$ (orange dots) and $|P_{225}-P_1|$ (violet dots), where $P_1, P_{100}$ and $P_{225}$ denote the probabilities in the 1-, 100-, and 225-voxel systems, respectively.}\label{fig:nonlinear.prob.dis}
\end{figure}
Our theoretical result is empirically verified in simulations of the following complex balanced network:
\begin{equation}
\varnothing\xrightarrow{k_1}X_1+2X_2\xrightarrow{k_2}X_2\xrightarrow{k_3}\varnothing.
\end{equation}
We consider three cases with different numbers of voxels in the system volume: 1 voxel (a well-mixed system), 100 voxels, and 225 voxels. The diffusion coefficients of the species in the 100-voxel system differ from those in the 225-voxel system. The steady-state distributions of species $X_1$ and $X_2$ are plotted in Fig.~\ref{fig:nonlinear.prob.dis}. As can be seen, the distributions of both species are consistent in all three cases.
From these result, it is pertinent to ask which conditions define a complex balanced network. Reference \cite{Feinberg.1995} proved that a weakly reversible reaction network with zero deficiency is a complex balanced network. This implies that in some cases, a complex balanced network can be identified by its network topology. 
In Ref.~\cite{Szederkenyi.2011}, complex balanced realizations of a given kinetic polynomial system were computed by a linear programming algorithm.

Thus far, we show that the steady-state distribution of a complex balanced network is a product of Poisson distributions. A network with such a distribution implies that the system is diffusion-independent at steady state. Therefore, we desire to know the condition under which the system establishes a Poisson-like steady-state distribution. This condition is embodied in the following theorem.
\begin{theorem}\label{theo:necessary.sufficient.condition}
The network possesses the steady-state distribution defined in Eq.~\eqref{eq:solution.prob.form} in all state spaces $\Gamma\subseteq\mbb{N}_{\geq 0}^{N}$ if and only if it is complex balanced.
\end{theorem}
\begin{proof}
We use the Fock space representation \cite{Peliti.1985} to describe the molecule-number changes of each species inside each voxel. A state vector $|\bm{n}\rangle$ with configuration $\bm{n}$ means that $n_{vi}$ molecules of species $X_i$ exist in voxel $v$. Using the annihilation and creation operators $a_{vi},a_{vi}^{\dag}$, i.e., $a_{vi}|n_{vi}\rangle=n_{vi}|n_{vi}-1\rangle,a_{vi}^{\dag}|n_{vi}\rangle=|n_{vi}+1\rangle$, we can map the probability distribution $P_{\Gamma}(\bm{n},t)$ to a state vector $|\psi(t)\rangle_{\Gamma}$, defined by
\begin{equation}\label{eq:generating.function.bosonic}
|\psi(t)\rangle_{\Gamma}=\sum_{\bm{n}}P_{\Gamma}(\bm{n},t)|\bm{n}\rangle=\sum_{\bm{n}}P_{\Gamma}(\bm{n},t)(\bm{a}^{\dag})^{\bm{n}}|0\rangle.
\end{equation}
This expression sums over all possible configurations $\bm{n}$ weighted by their occurrence probabilities at time $t$. To establish the time evolution of this state vector, we apply the master equation to obtain the following Schr{\"o}dinger equation:
\begin{equation}
\partial_t|\psi(t)\rangle_{\Gamma}=-\mca{H}(\bm{a}^{\dag},\bm{a})|\psi(t)\rangle_{\Gamma},\label{eq:schrodinger.equation}
\end{equation}
where $\mca{H}(\bm{a}^{\dag},\bm{a})$ represents the Hamiltonian action on the Fock space, expressed as shown in Appendix \ref{app.theorem.2}. In general, $\mca{H}(\bm{a}^{\dag},\bm{a})$ is the sum of several sub-actions created by each reaction of the system, e.g., a reaction of the form $\sum_{i=1}^{N}s_{ij}X_i^v\xrightarrow{k_j}\sum_{i=1}^{N}r_{ij}X_i^v$ yields a sub-action $k_j\omega^{1-\sum_{i=1}^{N}s_{ij}}\big(\prod_{i=1}^N(a_{vi}^{\dag})^{r_{ij}}-\prod_{i=1}^N(a_{vi}^{\dag})^{s_{ij}}\big)\prod_{i=1}^N(a_{vi})^{s_{ij}}$. The action $\mca{H}(\bm{a}^{\dag},\bm{a})$ is considered to be normally ordered, i.e., $a_{vi}^{\dag}$ is always to the left of $a_{vi}$. In a steady-state system, $\mca{H}(\bm{a}^{\dag},\bm{a})|\psi(t)\rangle_{\Gamma}=0$. Consequently, $\mca{H}(\bm{a}^{\dag},\bm{a})|\psi(t)\rangle$ is also $0$, where the state $|\psi(t)\rangle$ is defined as follows:
\begin{equation*}
\begin{aligned}
|\psi(t)\rangle&=\sum_{\Gamma}\frac{|\psi(t)\rangle_{\Gamma}}{\mca{N}_{\Gamma}}=\prod_{i=1}^{N}\prod_{v=1}^{|\mca{V}|}\sum_{n_{vi}\geq 0}\frac{(\omega c_ia_{vi}^{\dag})^{n_{vi}}}{n_{vi}!}|0\rangle\\
&=e^{\sum_{v,i}\omega c_ia_{vi}^{\dag}}|0\rangle.
\end{aligned}
\end{equation*}
In other words, the following condition
\begin{equation}
\mca{H}(\bm{a}^{\dag},\bm{a})e^{\sum_{v,i}\omega c_ia_{vi}^{\dag}}|0\rangle=0\label{eq:condition.hamilton}
\end{equation}
must hold. To derive a further condition with no involvement of $\bm{a}^{\dag}$ and $\bm{a}$, we consider the coherent states $|\phi_{vi}\rangle$ and $\langle\phi_{vi}|$, defining the right and left eigenstates of $a_{vi}$ and $a_{vi}^{\dag}$, respectively. Specifically, $a_{vi}|\phi_{vi}\rangle=\phi_{vi}|\phi_{vi}\rangle$ and $\langle\phi_{vi}|a_{vi}^{\dag}=\langle\phi_{vi}|\phi_{vi}^*$, with complex eigenvalue $\phi_{vi}\in\mbb{C}$. Multiplying both sides of Eq.~\eqref{eq:condition.hamilton} by the left coherent state $\langle\bm{\phi}|$, we obtain
\begin{equation}
0=\langle\bm{\phi}|\mca{H}(\bm{a}^{\dag},\bm{a})e^{\sum_{v,i}\omega c_ia_{vi}^{\dag}}|0\rangle\Leftrightarrow 0=\mca{H}(\bm{\phi}^*,\omega\widetilde{\bm{c}}),\label{eq:H.condition}
\end{equation}
where $\widetilde{\bm{c}}\in\mbb{R}^{|\mca{V}|N}$ is defined as $\widetilde{c}_{vi}=c_i$.
As $\mca{H}(\bm{\phi}^*,\omega\widetilde{\bm{c}})$ is a polynomial of $\bm{\phi}^*$, this result is possible only when the coefficients of all monomials are zero. Each reaction of the form $C_i\xrightarrow{a_{ii'}}C_{i'}$ in voxel $v$ contributes to $\mca{H}(\bm{\phi}^*,\omega\widetilde{\bm{c}})$ a quantity $\omega a_{ii'}\left(\Psi_{i'}(\bm{\phi}_v^*)-\Psi_i(\bm{\phi}_v^*)\right)\Psi_i(\bm{c})$. Therefore, by collecting the coefficients of $\Psi_i(\bm{\phi}_v^*)$ for each $i=1,\dots,M$ and $v\in\mca{V}$, we obtain the following relation:
\begin{equation}
\mca{H}(\bm{\phi}^*,\omega\widetilde{\bm{c}})=0,~\forall\bm{\phi}\in\mbb{C}^{|\mca{V}|N}\Leftrightarrow\bm{A}\cdot\Psi(\bm{c})=0,\label{eq:final.condition}
\end{equation}
meaning that the network is complex balanced at $\bm{c}$.
The details of these calculations are shown in Appendix \ref{app.theorem.2}.
From these results, we conclude that the necessary and sufficient condition for a steady-state distribution (Eq.~\eqref{eq:solution.prob.form}) is that the network is complex balanced.
\end{proof}
We note that the sufficient condition of Theorem \ref{theo:necessary.sufficient.condition} has been studied in Ref.~\cite{Lubensky.2010} (i.e., if the network is complex balanced, then the steady-state distribution has a form as in Eq.~\eqref{eq:solution.prob.form}).
Above we investigate the relation between diffusion and the distributions of the reactant species in steady state. We now present another result that holds under non-steady-state conditions.
\begin{theorem}\label{theo:linear.networks}
When the reaction network is linear, the RDME can be reduced to the CME. Equivalently, the diffusion can be ignored in such case.
\end{theorem}
\begin{proof}
The system volume $\Omega$ is related to the voxel volume $\omega$ as $\Omega=|\mca{V}|\omega$. Now, for each state vector $\widehat{\bm{n}}=[\widehat{n}_1,\widehat{n}_2,\dots,\widehat{n}_N]^{\top}\in\mbb{N}_{\geq 0}^{N}$ representing the number of molecules of the reactant species, i.e., $\widehat{n}_i$ is the total number of molecules of species $X_i$ in the system, we define the set $\mca{S}(\widehat{\bm{n}})=\big\{\bm{n}\in\mbb{N}_{\geq 0}^{|\mca{V}|N} \mid \sum_{v\in \mca{V}}\bm{n}_{v}=\widehat{\bm{n}}\big\}$. Let $P(\widehat{\bm{n}},t)$ be the probability of the system being in state $\widehat{\bm{n}}$ at time $t$. In terms of $P(\bm{n},t)$, this probability becomes $P(\widehat{\bm{n}},t)=\sum_{\bm{n}\in \mca{S}(\widehat{\bm{n}})}P(\bm{n},t)$. As $\sum_{\widehat{\bm{n}}}P(\widehat{\bm{n}},t)=\sum_{\bm{n}}P(\bm{n},t)=1$, $P(\widehat{\bm{n}},t)$ is a probability distribution. To show that this probability distribution satisfies the CME given by Eq.~\eqref{eq:cme}, we calculate the time derivative of $P(\widehat{\bm{n}},t)$ as follows:
\begin{equation}\label{eq:time.deriv}
\partial_tP(\widehat{\bm{n}},t)=\sum_{\bm{n}\in \mca{S}(\widehat{\bm{n}})}\partial_tP(\bm{n},t).
\end{equation}
Substituting Eq.~\eqref{eq:rdme} into the right-hand side of Eq.~\eqref{eq:time.deriv}, we obtain an equation with both diffusion and reaction terms on the right. After some algebraic transformations, the diffusion term disappears and only the reaction term remains (see Appendix \ref{app.theorem.3}). As the reaction network is linear, i.e., $\sum_{i}s_{ij}\leq 1~\forall~j=1,\dots,K$, the propensity function $f_{vj}(\bm{n},\omega)$ must be one of two forms: $f_{vj}(\bm{n},\omega)=k_jn_{vi}$ or $k_j\omega$. Substituting the exact form of each propensity function into Eq.~\eqref{eq:time.deriv}, we finally obtain the following master equation for $P(\widehat{\bm{n}},t)$:
\begin{equation}\label{eq:rdme.derived.cme}
\partial_tP(\widehat{\bm{n}},t)=\sum_{j=1}^{K}\left(\bm{\mbb{E}}^{-\bm{V}_j}-1\right)f_j(\widehat{\bm{n}},\Omega)P(\widehat{\bm{n}},t).
\end{equation}
Obviously, this differential equation is identical to the CME stated in Eq.~\eqref{eq:cme}, and contains no diffusion factors. Therefore, it can be concluded that diffusion can be ignored in linear reaction networks.
\end{proof}
In Theorem \ref{theo:linear.networks}, we demonstrate that the RDME reduces to the CME in the case of linear reaction networks, which implies that diffusion does not affect the stochastic dynamics of the system at an arbitrary time. From the view of the Smoluchowski model, this statement appears to be obvious.
By regarding the network as interacting many-particle system and introducing diffusion and reaction operators \cite{Doi.1976.1}, the same result can be derived.
However, it is not evident from the view of the RDME. The agreement of results in these different models serves as validation for the RDME.
\begin{figure}[t]
	\centering
	\includegraphics[width=0.47\textwidth]{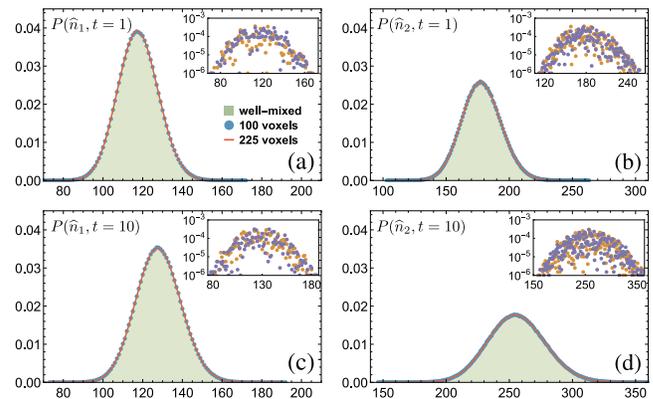}
	\caption{Probability distributions (a) $P(\widehat{n}_1,t=1)$, (b) $P(\widehat{n}_2,t=1)$, (c) $P(\widehat{n}_1,t=10)$, and (d) $P(\widehat{n}_2,t=10)$ of two species $X_1,X_2$ at times $t=1$ (upper panels) and $t=10$ (lower panels) of linear reaction network. Each panel shows the distributions of the 1-voxel system (green region), 100-voxel system (blue dots), and 225-voxel system (red line). The parameters are $k_1=1,k_2=1,k_3=2,k_4=1,\Omega=128$. The diffusion rates of species $X_1,X_2$ are $d_1=1,d_2=2$ (100 voxels) and $d_1=2,d_2=1$ (225 voxels). Insets show the absolute probability differences $|P_{100}-P_1|$ (orange dots) and $|P_{225}-P_1|$ (violet dots), where $P_1, P_{100}$ and $P_{225}$ indicate the probability in the 1-, 100-, and 225-voxel systems, respectively.}\label{fig:linear.prob.dis}
\end{figure}

We numerically verify the result of Theorem \ref{theo:linear.networks} on a simple linear reaction network, namely, a coarse-grained model of enzymatic reactions and gene expressions. The network consists of two reactant species $X_1$ and $X_2$ and four reactions \cite{Nicola.2006,Tostevin.2010}:
\begin{equation}
\varnothing\mathrel{\mathop{\rightleftarrows}_{k_2}^{k_1}}X_1,X_1\xrightarrow{k_3}X_1+X_2,X_2\xrightarrow{k_4}\varnothing.
\end{equation}
Again, we divide the cell volume into 1, 100, and 225 voxels with different diffusion coefficients of $X_1$ and $X_2$. The result is displayed in Fig.~\ref{fig:linear.prob.dis}. As before, the distributions of each species at times $t=1$ and $t=10$ are identical in all three cases. These numerical results empirically validate Theorem \ref{theo:linear.networks}.

\section{Conclusions}
In summary, within the approximation of the RDME, we proved that diffusion in complex-balanced networks does not affect the steady-state distribution of the system. We also showed that a diffusion-included reaction network has a Poisson-like steady-state distribution if and only if it is complex balanced, analogously to the well-mixed case described in \cite{Daniele.2016}.
Moreover, we demonstrated that the RDME can be reduced to the CME in the case of linear reaction networks.
These results help to clarify the conditions under which diffusion is negligible. Under such conditions, the system can be described by the CME instead of the intractable RDME. In nonlinear networks that are not complex-balanced, how diffusion affects the stochastic system dynamics, or whether it can be ignored, requires further investigation.

It appears that functional biological networks satisfying the complex balanced condition are not widespread in real-world systems.
Nevertheless, weakly reversible networks have been successfully applied in modeling signal transduction pathways \cite{Friedmann.2013} and asymmetric stem-cell division \cite{Mayer.2005}.
Besides that, complex balanced networks whose Fano factor is equal to one can be used in analyzing and approximating cascade networks or metabolic pathways, wherein noise is not propagated from upstream to downstream \cite{Levine.2007}.
Although our results are obtained with the approximation of the RDME, it is expected that the derived diffusion-dynamics laws provide suggestive results for real physical systems.

\section*{Acknowledgment}
This work was supported by MEXT KAKENHI Grant No.~JP16K00325.

\appendix

\section{Detailed calculations in Theorem 1}\label{app.theorem.1}
\subsection{Detailed calculations of Equation~\eqref{eq:diffusion.term}}
The detailed calculation of Eq.~\eqref{eq:diffusion.term} is given below:
\begin{widetext}
\begin{align*}
&\sum_{v\in \mca{V}}\sum_{v'\in  N_e(v)}\sum_{i=1}^{N}\left(\bm{\mrm{E}}^{\bm{1}_{vi}-\bm{1}_{v'i}}-1\right)d_{i}n_{vi}P_{\Gamma}(\bm{n},t)\\
&=\sum_{v\in \mca{V}}\sum_{v'\in  N_e(v)}\sum_{i=1}^{N}\bigg(d_{i}(n_{vi}+1)P_{\Gamma}(\bm{n}+\bm{1}_{vi}-\bm{1}_{v'i},t)-d_{i}n_{vi}P_{\Gamma}(\bm{n},t)\bigg)\\
&=\sum_{i=1}^{N}d_i\sum_{v\in \mca{V}}\sum_{v'\in  N_e(v)}\bigg(\mca{N}_{\Gamma}(n_{vi}+1)\frac{(\omega c_i)^{n_{vi}+1}}{(n_{vi}+1)!}\frac{(\omega c_i)^{n_{v'i}-1}}{(n_{v'i}-1)!}\prod_{i'\neq i}\frac{(\omega c_{i'})^{n_{vi'}}(\omega c_{i'})^{n_{v'i'}}}{n_{vi'}!n_{v'i'}!}\prod_{\widetilde{v}\neq v,v'}P_{\Gamma}^*(\bm{n}_{\widetilde{v}},t)-n_{vi}P_{\Gamma}(\bm{n},t)\bigg)\\
&=\sum_{i=1}^{N}d_i\sum_{v\in \mca{V}}\sum_{v'\in  N_e(v)}\bigg(\mca{N}_{\Gamma}\frac{(\omega c_i)^{n_{vi}}}{n_{vi}!}\frac{(\omega c_i)^{n_{v'i}}}{(n_{v'i}-1)!}\prod_{i'\neq i}\frac{(\omega c_{i'})^{n_{vi'}}(\omega c_{i'})^{n_{v'i'}}}{n_{vi'}!n_{v'i'}!}\prod_{\widetilde{v}\neq v,v'}P_{\Gamma}^*(\bm{n}_{\widetilde{v}},t)-n_{vi}P_{\Gamma}(\bm{n},t)\bigg)\\
&=\sum_{i=1}^{N}d_i\sum_{v\in \mca{V}}\sum_{v'\in  N_e(v)}\bigg(n_{v'i}P_{\Gamma}(\bm{n},t)-n_{vi}P_{\Gamma}(\bm{n},t)\bigg)\\
&=0.
\end{align*}
\end{widetext}
\subsection{Detailed calculations of Equation~\eqref{eq:reaction.term}}
We here reveal the details of Eq.~\eqref{eq:reaction.term}.
The equation $\bm{A}\cdot\Psi(\bm{c})=0$ means that $\sum_{C_i\to C_{i'}}a_{ii'}\Psi_{i}(\bm{c})-\sum_{C_{i'}\to C_i}a_{i'i}\Psi_{i'}(\bm{c})=0$ for each complex $C_{i'}\in\mca{C}$. The left side of Eq.~\eqref{eq:reaction.term} can be transformed as
\begin{widetext}
\begin{align*}
&\sum_{j=1}^{K}\left(\bm{\mbb{E}}^{-\bm{V}_{j}}-1\right)f_{j}(\bm{n}_v,\omega)P_{\Gamma}^*(\bm{n}_v,t)\\
&=\sum_{j=1}^{K}\left(f_j(\bm{n}_v-\bm{V}_j,\omega)P_{\Gamma}^*(\bm{n}_v-\bm{V}_j,t)-f_{j}(\bm{n}_v,\omega)P_{\Gamma}^*(\bm{n}_v,t)\right)\\
&=\sum_{C_i\to C_{i'}}a_{ii'}\omega^{1-\sum_{k=1}^{N}y_{ki}}\left(\prod_{k=1}^N\frac{(n_{vk}+y_{ki}-y_{ki'})!}{(n_{vk}-y_{ki'})!}\frac{(\omega c_k)^{n_{vk}+y_{ki}-y_{ki'}}}{(n_{vk}+y_{ki}-y_{ki'})!}-\prod_{k=1}^N\frac{n_{vk}!}{(n_{vk}-y_{ki})!}\frac{(\omega c_k)^{n_{vk}}}{n_{vk}!}\right)\\
&=\sum_{C_i\to C_{i'}}a_{ii'}\omega^{1-\sum_{k=1}^{N}y_{ki}}\left(\prod_{k=1}^N\frac{(\omega c_k)^{n_{vk}+y_{ki}-y_{ki'}}}{(n_{vk}-y_{ki'})!}-\prod_{k=1}^N\frac{(\omega c_k)^{n_{vk}}}{(n_{vk}-y_{ki})!}\right)\\
&=\omega\sum_{C_i\to C_{i'}}\left(a_{ii'}\frac{\Psi_{i}(\bm{c})}{\Psi_{i'}(\bm{c})}\prod_{k=1}^N\frac{\omega^{n_{vk}-y_{ki'}}c_k^{n_{vk}}}{(n_{vk}-y_{ki'})!}-a_{ii'}\prod_{k=1}^N\frac{\omega^{n_{vk}-y_{ki}}c_k^{n_{vk}}}{(n_{vk}-y_{ki})!}\right)\\
&=\omega\sum_{C_{i'}\in\mca{C}}\left(\sum_{C_i\to C_{i'}}a_{ii'}\frac{\Psi_{i}(\bm{c})}{\Psi_{i'}(\bm{c})}\prod_{k=1}^N\frac{\omega^{n_{vk}-y_{ki'}}c_k^{n_{vk}}}{(n_{vk}-y_{ki'})!}-\sum_{C_{i'}\to C_i}a_{i'i}\prod_{k=1}^N\frac{\omega^{n_{vk}-y_{ki'}}c_k^{n_{vk}}}{(n_{vk}-y_{ki'})!}\right)\\
&=\omega\sum_{C_{i'}\in\mca{C}}\frac{1}{\Psi_{i'}(\bm{c})}\prod_{k=1}^N\frac{\omega^{n_{vk}-y_{ki'}}c_k^{n_{vk}}}{(n_{vk}-y_{ki'})!}\left(\sum_{C_i\to C_{i'}}a_{ii'}\Psi_{i}(\bm{c})-\sum_{C_{i'}\to C_i}a_{i'i}\Psi_{i'}(\bm{c})\right)\\
&=0.
\end{align*}
\end{widetext}
The same result (but omitting the details) is given in Ref. \cite{Anderson.2010}.

\subsection{Detailed calculations of Equation~\eqref{eq:explicit.distribution}}
Finally, we compute the explicit form of the distribution $P_{\Gamma}(\widehat{\bm{n}})$ in Eq.~\eqref{eq:explicit.distribution}. We have
\begin{align*}
P_{\Gamma}(\widehat{\bm{n}})&=\sum_{\bm{n}:\sum_{v}\bm{n}_v=\widehat{\bm{n}}}P_{\Gamma}(\bm{n},t)\\
&=\mca{N}_{\Gamma}\sum_{\substack{\bm{n}_1,\dots,\bm{n}_{|\mca{V}|}\in\mbb{N}_{\geq 0}^{N}\\\sum_{v}\bm{n}_v=\widehat{\bm{n}}}}\prod_{v\in\mca{V}}\prod_{i=1}^{N}\frac{(\omega c_i)^{n_{vi}}}{n_{vi}!}\\
&=\mca{N}_{\Gamma}\prod_{i=1}^{N}\Bigg\{\sum_{\substack{n_{1i},\dots,n_{|\mca{V}|i}\geq 0\\\sum_{v\in\mca{V}}n_{vi}=\widehat{n}_i}}\prod_{v\in\mca{V}}\frac{(\omega c_i)^{n_{vi}}}{n_{vi}!}\Bigg\}\\
&=\mca{N}_{\Gamma}\prod_{i=1}^{N}\Bigg\{\frac{(\omega c_i)^{\widehat{n}_i}}{\widehat{n}_i!}\sum_{\substack{n_{1i},\dots,n_{|\mca{V}|i}\geq 0\\\sum_{v\in\mca{V}}n_{vi}=\widehat{n}_i}}\frac{\left(\sum_{v\in\mca{V}}n_{vi}\right)!}{\prod_{v\in\mca{V}}n_{vi}!}\Bigg\}\\
&=\mca{N}_{\Gamma}\prod_{i=1}^{N}\frac{(\omega c_i)^{\widehat{n}_i}}{\widehat{n}_i!}|\mca{V}|^{\widehat{n}_i}=\mca{N}_{\Gamma}\prod_{i=1}^{N}\frac{(\Omega c_i)^{\widehat{n}_i}}{\widehat{n}_i!}.
\end{align*}
In transforming the fourth equation to the fifth one, we exploited the following equality:
\begin{equation*}
\sum_{\substack{x_1,\dots,x_m\geq0\\\sum_{i=1}^{m}x_i=n}}\frac{(x_1+\dots+x_m)!}{x_1!\dots x_m!}=m^n,~\forall m\in\mathbb{N}_{>0},n\in\mathbb{N}_{\geq 0}.
\end{equation*}

\section{Detailed calculations in Theorem 2}\label{app.theorem.2}
Before presenting the calculations, we state several properties of the bosonic operators $a_{vi}^{\dag}$ and $a_{vi}$.
\begin{align*}
|n_{vi}\rangle&=(a_{vi}^{\dag})^{n_{vi}}|0\rangle,\\
(a_{vi})^l(a_{vi}^{\dag})^{k}|n_{vi}\rangle&=\prod_{j=0}^{l-1}(n_{vi}+k-j)|n_{vi}+k-l\rangle,\\
(a_{vi}^{\dag})^{k}(a_{vi})^l|n_{vi}\rangle&=\prod_{j=0}^{l-1}(n_{vi}-j)|n_{vi}+k-l\rangle,\\
[a_{vi},a_{v'i'}^{\dag}]&=a_{vi}a_{v'i'}^{\dag}-a_{v'i'}^{\dag}a_{vi}=\delta_{vv'}\delta_{ii'},\\
[a_{vi}^{\dag},a_{v'i'}^{\dag}]&=[a_{vi},a_{v'i'}]=0.
\end{align*}
For a general configuration $\bm{n}$, we define the corresponding state vector $|\bm{n}\rangle$ as
\begin{equation}
|\bm{n}\rangle=(\bm{a}^{\dag})^{\bm{n}}|0\rangle=\prod_{v\in\mca{V}}\prod_{i=1}^N(a_{vi}^{\dag})^{n_{vi}}|0\rangle.
\end{equation}
For convenience, we note that
\begin{align}
	e^{ca_{vi}}f(a_{vi}^{\dag})&=f(a_{vi}^{\dag}+c)e^{ca_{vi}},\label{eq:prop.1} \\
	e^{ca_{vi}^{\dag}}f(a_{vi})&=f(a_{vi}-c)e^{ca_{vi}^{\dag}},\label{eq:prop.2}
\end{align}
where $c\in\mathbb{C}$ is a complex number and $f$ is an arbitrary function.

\subsection{Detailed calculations of Equation~\eqref{eq:schrodinger.equation}}
We first derive the explicit form of the Hamiltonian action $\mca{H}(\bm{a}^{\dag},\bm{a})$ in Eq.~\eqref{eq:schrodinger.equation}. Suppose that the network contains a set $\mca{R}$ of reactions $R_j$ of the general form $\sum_{v\in\mca{V}}\sum_{i=1}^Np_{vi}^jX_{i}^v\xrightarrow{k_j}\sum_{v\in\mca{V}}\sum_{i=1}^Nq_{vi}^jX_{i}^v$, where $p_{vi}^j$ and $q_{vi}^j$ are the stoichiometric coefficients. For each reaction $R_j$, we define a stoichiometric vector $\bm{V}^j\in\mbb{Z}^{|\mca{V}|N}$ as $V_{vi}^j=q_{vi}^j-p_{vi}^j$. Starting from the master equation, we have
\begin{widetext}
\begin{align*}
\partial_t|\psi(t)\rangle_{\Gamma}&=\sum_{\bm{n}}\partial_tP_{\Gamma}(\bm{n},t)(\bm{a}^{\dag})^{\bm{n}}|0\rangle\\
&=\sum_{\bm{n}}\sum_{R_j\in\mca{R}}k_j\omega^{1-\sum_{v,i}p_{vi}^j}\left[\prod_{v,i}\frac{(n_{vi}+p_{vi}^j-q_{vi}^j)!}{(n_{vi}-q_{vi}^j)!}P_{\Gamma}(\bm{n}-\bm{V}^j,t)-\prod_{v,i}\frac{n_{vi}!}{(n_{vi}-p_{vi}^j)!}P_{\Gamma}(\bm{n},t)\right](\bm{a}^{\dag})^{\bm{n}}|0\rangle.
\end{align*}
Note that the two terms inside the bracket can be obtained using operators as follows:
\begin{align*}
\prod_{v,i}\frac{(n_{vi}+p_{vi}^j-q_{vi}^j)!}{(n_{vi}-q_{vi}^j)!}P_{\Gamma}(\bm{n}-\bm{V}^j,t)(\bm{a}^{\dag})^{\bm{n}}|0\rangle&=\prod_{v,i}(a_{vi}^{\dag})^{q_{vi}^j}(a_{vi})^{p_{vi}^j}P_{\Gamma}(\bm{n}-\bm{V}^j,t)(\bm{a}^{\dag})^{\bm{n}-\bm{V}^j}|0\rangle,\\
\prod_{v,i}\frac{n_{vi}!}{(n_{vi}-p_{vi}^j)!}P_{\Gamma}(\bm{n},t)(\bm{a}^{\dag})^{\bm{n}}|0\rangle&=\prod_{v,i}(a_{vi}^{\dag})^{p_{vi}^j}(a_{vi})^{p_{vi}^j}P_{\Gamma}(\bm{n},t)(\bm{a}^{\dag})^{\bm{n}}|0\rangle.
\end{align*}
Using these equalities, $\partial_t|\psi(t)\rangle_{\Gamma}$ is calculated as follows:
\begin{align*}
&\partial_t|\psi(t)\rangle_{\Gamma}\\
&=\sum_{\bm{n}}\sum_{R_j\in\mca{R}}k_j\omega^{1-\sum_{v,i}p_{vi}^j}\left[\prod_{v,i}(a_{vi}^{\dag})^{q_{vi}^j}(a_{vi})^{p_{vi}^j}P_{\Gamma}(\bm{n}-\bm{V}^j,t)(\bm{a}^{\dag})^{\bm{n}-\bm{V}^j}|0\rangle-\prod_{v,i}(a_{vi}^{\dag})^{p_{vi}^j}(a_{vi})^{p_{vi}^j}P_{\Gamma}(\bm{n},t)(\bm{a}^{\dag})^{\bm{n}}|0\rangle\right]\\
&=\sum_{R_j\in\mca{R}}k_j\omega^{1-\sum_{v,i}p_{vi}^j}\left[\prod_{v,i}(a_{vi}^{\dag})^{q_{vi}^j}(a_{vi})^{p_{vi}^j}\sum_{\bm{n}}P_{\Gamma}(\bm{n}-\bm{V}^j,t)(\bm{a}^{\dag})^{\bm{n}-\bm{V}^j}|0\rangle-\prod_{v,i}(a_{vi}^{\dag})^{p_{vi}^j}(a_{vi})^{p_{vi}^j}\sum_{\bm{n}}P_{\Gamma}(\bm{n},t)(\bm{a}^{\dag})^{\bm{n}}|0\rangle\right]\\
&=\sum_{R_j\in\mca{R}}k_j\omega^{1-\sum_{v,i}p_{vi}^j}\left[\prod_{v,i}(a_{vi}^{\dag})^{q_{vi}^j}(a_{vi})^{p_{vi}^j}-\prod_{v,i}(a_{vi}^{\dag})^{p_{vi}^j}(a_{vi})^{p_{vi}^j}\right]|\psi(t)\rangle_{\Gamma}.
\end{align*}
Thus, the general form of $\mca{H}$ is obtained as
\begin{equation}
\mca{H}(\bm{a}^{\dag},\bm{a})=\sum_{R_j\in\mca{R}}k_j\omega^{1-\sum_{v,i}p_{vi}^j}\left[\prod_{v,i}(a_{vi}^{\dag})^{q_{vi}^j}-\prod_{v,i}(a_{vi}^{\dag})^{p_{vi}^j}\right]\prod_{v,i}(a_{vi})^{p_{vi}^j}.
\end{equation}
For a diffusion-included reaction network involving the following reactions
\begin{align*}
&s_{1j}X_1^v+\dots+s_{Nj}X_N^v\xrightarrow{k_j}r_{1j}X_1^v+\dots+r_{Nj}X_N^v,\\
&X_{i}^{v}\xrightarrow{d_i}X_{i}^{v'},~\forall~1\leq i\leq N,v\in\mca{V},v'\in N_e(v),
\end{align*}
the Hamiltonian action $\mca{H}(\bm{a}^{\dag},\bm{a})$ in Eq.~\eqref{eq:schrodinger.equation} takes the following form
\begin{equation}
\mca{H}(\bm{a}^{\dag},\bm{a})=\sum_{j=1}^{K}\sum_{v\in\mca{V}}k_j\omega^{1-\sum_{i=1}^{N}s_{ij}}\left[\prod_{i=1}^N(a_{vi}^{\dag})^{r_{ij}}-\prod_{i=1}^{N}(a_{vi}^{\dag})^{s_{ij}}\right]\prod_{i=1}^N(a_{vi})^{s_{ij}}+\sum_{i=1}^N\sum_{v\in\mca{V}}\sum_{v'\in N_e(v)}d_i(a_{v'i}^{\dag}-a_{vi}^{\dag})a_{vi}.
\end{equation}
\end{widetext}
We note that this form of $\mca{H}$ is already normal-ordered.

\subsection{Detailed calculations of Equation~\eqref{eq:H.condition}}
Equation~\eqref{eq:H.condition} is derived through the following steps:
\begin{align}
0&=\langle\bm{\phi}|\mca{H}(\bm{a}^{\dag},\bm{a})e^{\sum_{v,i}\omega c_ia_{vi}^{\dag}}|0\rangle\\
\Leftrightarrow 0&=\langle\bm{\phi}|e^{\sum_{v,i}\omega c_ia_{vi}^{\dag}}\mca{H}(\bm{a}^{\dag},\bm{a}+\omega\widetilde{\bm{c}})|0\rangle\label{eq:eq17.2} \\
\Leftrightarrow 0&=e^{\sum_{v,i}\omega c_i\phi_{vi}^{*}}\langle\bm{\phi}|\mathcal{H}(\bm{a}^{\dag},\bm{a}+\omega\widetilde{\bm{c}})|0\rangle\\
\Leftrightarrow 0&=e^{\sum_{v,i}\omega c_i\phi_{vi}^{*}}\langle\bm{\phi}|\mathcal{H}(\bm{\phi}^{*},\omega\widetilde{\bm{c}})|0\rangle\label{eq:eq17.4} \\
\Leftrightarrow 0&=e^{\sum_{v,i}\omega c_i\phi_{vi}^{*}}\mathcal{H}(\bm{\phi}^{*},\omega\widetilde{\bm{c}})\langle\bm{\phi}|0\rangle\\
\Leftrightarrow 0&=\mca{H}(\bm{\phi}^*,\omega\widetilde{\bm{c}}).\label{eq:eq17.6}
\end{align}
In Eq.~\eqref{eq:eq17.2}, we use the property stated in Eq.~\eqref{eq:prop.2}. In Eq.~\eqref{eq:eq17.4}, the operators $a_{vi}~(v\in\mca{V},i=1,\dots,N)$ are absorbed into $|0\rangle~(\because a_{vi}|0\rangle=0)$, and the operators $a_{vi}^{\dag}$ are replaced by $\phi_{vi}^*~(\because \langle\phi_{vi}|a_{vi}^{\dag}=\langle\phi_{vi}|\phi_{vi}^*)$. The result Eq.~\eqref{eq:eq17.6} is obtained by noting that
\begin{align*}
e^{\sum_{v,i}\omega c_i\phi_{vi}^{*}}&\neq 0,\\
\langle\bm{\phi}|0\rangle=\prod_{v,i}\langle\phi_{vi}|0\rangle=\prod_{v,i}e^{-\frac{1}{2}|\phi_{vi}|^2}&\neq 0.
\end{align*}

\subsection{Detailed calculations of Equation~\eqref{eq:final.condition}}
Equation~\eqref{eq:final.condition} is given by
\begin{equation*}
\mca{H}(\bm{\phi}^*,\omega\widetilde{\bm{c}})=0,~\forall\bm{\phi}\in\mbb{C}^{|\mca{V}|N}\Leftrightarrow\bm{A}\cdot\Psi(\bm{c})=0.
\end{equation*}
The Hamiltonian action $\mca{H}(\bm{a}^{\dag},\bm{a})$ is described by
\begin{widetext}
\begin{equation*}
\mca{H}(\bm{a}^{\dag},\bm{a})=\sum_{j=1}^{K}\sum_{v\in\mca{V}}k_j\omega^{1-\sum_{i=1}^{N}s_{ij}}\left[\prod_{i=1}^N(a_{vi}^{\dag})^{r_{ij}}-\prod_{i=1}^{N}(a_{vi}^{\dag})^{s_{ij}}\right]\prod_{i=1}^N(a_{vi})^{s_{ij}}+\sum_{i=1}^N\sum_{v\in\mca{V}}\sum_{v'\in N_e(v)}d_i(a_{v'i}^{\dag}-a_{vi}^{\dag})a_{vi}.
\end{equation*}
The case $\mca{H}(\bm{\phi}^*,\omega\widetilde{\bm{c}})=0$ is equivalent to
\begin{align*}
&\sum_{j=1}^{K}\sum_{v\in\mca{V}}k_j\omega^{1-\sum_{i=1}^{N}s_{ij}}\left[\prod_{i=1}^N(\phi_{vi}^{*})^{r_{ij}}-\prod_{i=1}^{N}(\phi_{vi}^{*})^{s_{ij}}\right]\prod_{i=1}^N(\omega\widetilde{c}_{vi})^{s_{ij}}+\sum_{i=1}^N\sum_{v\in\mca{V}}\sum_{v'\in N_e(v)}d_i(\phi_{v'i}^{*}-\phi_{vi}^{*})\omega\widetilde{c}_{vi}=0\\
&\Leftrightarrow\sum_{C_i\to C_{i'}}\sum_{v\in\mca{V}}\omega a_{ii'}\left[\Psi_{i'}(\bm{\phi}_v^*)-\Psi_{i}(\bm{\phi}_v^*)\right]\Psi_{i}(\bm{c})+\sum_{i=1}^Nd_i\sum_{\substack{v,v'\in\mca{V}\\v'\in N_e(v)\\ v\in N_e(v')}}\left[(\phi_{v'i}^{*}-\phi_{vi}^{*})\omega\widetilde{c}_{vi}+(\phi_{vi}^{*}-\phi_{v'i}^{*})\omega\widetilde{c}_{v'i}\right]=0\\
&\Leftrightarrow\omega\sum_{v\in\mca{V}}\sum_{C_i\in\mca{C}}\Psi_{i}(\bm{\phi}_{v}^*)\left[\sum_{C_{i'}\to C_{i}}a_{i'i}\Psi_{i'}(\bm{c})-\sum_{C_{i}\to C_{i'}}a_{ii'}\Psi_{i}(\bm{c})\right]=0,~\forall\bm{\phi}\in\mbb{C}^{|\mca{V}|N}\\
&\Leftrightarrow \sum_{C_{i'}\to C_{i}}a_{i'i}\Psi_{i'}(\bm{c})-\sum_{C_{i}\to C_{i'}}a_{ii'}\Psi_{i}(\bm{c})=0,~\forall C_i\in\mca{C}\\
&\Leftrightarrow \bm{A}\cdot\Psi(\bm{c})=0.
\end{align*}
\end{widetext}

\section{Detailed calculations in Theorem 3}\label{app.theorem.3}
The master equation of $P(\widehat{\bm{n}},t)$ is derived as follows:
\begin{widetext}
\begin{equation}\label{eq:eq.of.P.hat}
\begin{aligned}
\partial_tP(\widehat{\bm{n}},t)&=\sum_{\bm{n}\in \mca{S}(\widehat{\bm{n}})}\partial_tP(\bm{n},t)\\
&=\sum_{\bm{n}\in\mca{S}(\widehat{\bm{n}})}\bigg(\sum_{v\in \mca{V}}\sum_{v'\in N_e(v)}\sum_{i=1}^{N}\bigg(d_{i}(n_{vi}+1)P(\bm{n}+\bm{1}_{vi}-\bm{1}_{v'i},t)-d_{i}n_{vi}P(\bm{n},t)\bigg)\\
&+\sum_{v\in \mca{V}}\sum_{j=1}^{K}\bigg(f_{vj}(\bm{n}-\widetilde{\bm{V}}_{vj},\omega)P(\bm{n}-\widetilde{\bm{V}}_{vj},t)-f_{vj}(\bm{n},\omega)P(\bm{n},t)\bigg)\bigg).
\end{aligned}
\end{equation}
As $\bm{n}\in\mca{S}(\widehat{\bm{n}})\rightarrow\bm{n}+\bm{1}_{vi}-\bm{1}_{v'i}=\bm{\widetilde{n}}\in\mca{S}(\widehat{\bm{n}})$, the first term of the right-hand side in Eq.~\eqref{eq:eq.of.P.hat} becomes
\begin{equation*}
\begin{aligned}
&\sum_{\bm{n}\in\mca{S}(\widehat{\bm{n}})}\sum_{v\in\mca{V}}\sum_{v'\in N_e(v)}\sum_{i=1}^{N}\bigg(d_{i}(n_{vi}+1)P(\bm{n}+\bm{1}_{vi}-\bm{1}_{v'i},t)-d_{i}n_{vi}P(\bm{n},t)\bigg)\\
&=\sum_{\bm{n}\in\mca{S}(\widehat{\bm{n}})}\sum_{v\in\mca{V}}\sum_{v'\in N_e(v)}\sum_{i=1}^{N}d_{i}(n_{vi}+1)P(\bm{n}+\bm{1}_{vi}-\bm{1}_{v'i},t)-\sum_{\bm{n}\in\mca{S}(\widehat{\bm{n}})}\sum_{v\in\mca{V}}\sum_{v'\in N_e(v)}\sum_{i=1}^{N}d_{i}n_{vi}P(\bm{n},t)\\
&=\sum_{\bm{\widetilde{n}}\in\mca{S}(\widehat{\bm{n}})}\sum_{v\in\mca{V}}\sum_{v'\in N_e(v)}\sum_{i=1}^{N}d_{i}\widetilde{n}_{vi}P(\bm{\widetilde{n}},t)-\sum_{\bm{n}\in\mca{S}(\widehat{\bm{n}})}\sum_{v\in\mca{V}}\sum_{v'\in N_e(v)}\sum_{i=1}^{N}d_{i}n_{vi}P(\bm{n},t)\\
&=0.
\end{aligned}
\end{equation*}

As the reaction network is linear, the propensity function $f_{vj}(\bm{n},\omega)$ takes one of two forms: $f_{vj}(\bm{n},\omega)=k_jn_{vi}$ or $f_{vj}(\bm{n},\omega)=k_j\omega$, where $k_j$ is the reaction rate and $i$ is the index of some species.
When $f_{vj}(\bm{n},\omega)=k_jn_{vi}$, the second term can be transformed as follows:
\begin{align*}
	&\sum_{\bm{n}\in\mca{S}(\widehat{\bm{n}})}\sum_{v\in \mca{V}}\bigg(f_{vj}(\bm{n}-\widetilde{\bm{V}}_{vj},\omega)P(\bm{n}-\widetilde{\bm{V}}_{vj},t)-f_{vj}(\bm{n},\omega)P(\bm{n},t)\bigg)\\
	&=k_j\sum_{\bm{n}\in\mca{S}(\widehat{\bm{n}})}\sum_{v\in \mca{V}}\bigg((n_{vi}-V_{ij})P(\bm{n}-\widetilde{\bm{V}}_{vj},t)-n_{vi}P(\bm{n},t)\bigg)\\
	&=k_j\bigg(\sum_{\bm{n}\in\mca{S}(\widehat{\bm{n}})}\sum_{v\in \mca{V}}(n_{vi}-V_{ij})P(\bm{n}-\widetilde{\bm{V}}_{vj},t)-\sum_{\bm{n}\in\mca{S}(\widehat{\bm{n}})}\sum_{v\in \mca{V}}n_{vi}P(\bm{n},t)\bigg)\\
	&=k_j\bigg(\sum_{\bm{n}\in\mca{S}(\widehat{\bm{n}})}\sum_{v\in \mca{V}}n_{vi}P(\bm{n}-\widetilde{\bm{V}}_{vj},t)-\sum_{\bm{n}\in\mca{S}(\widehat{\bm{n}})}\sum_{v\in \mca{V}}V_{ij}P(\bm{n}-\widetilde{\bm{V}}_{vj},t)-\sum_{\bm{n}\in\mca{S}(\widehat{\bm{n}})}\widehat{n}_{i}P(\bm{n},t)\bigg)\\
	&=k_j\bigg(\sum_{\bm{\widetilde{n}}\in \mca{S}(\widehat{\bm{n}}-\bm{V}_j)}\sum_{v\in \mca{V}}(\widetilde{n}_{vi}+V_{ij})P(\bm{\widetilde{n}},t)-V_{ij}\sum_{v\in \mca{V}}\sum_{\bm{n}\in\mca{S}(\widehat{\bm{n}})}P(\bm{n}-\widetilde{\bm{V}}_{vj},t)-\widehat{n}_{i}P(\widehat{\bm{n}},t)\bigg)\\
	&=k_j\bigg((\widehat{n}_i-V_{ij})P(\widehat{\bm{n}}-\bm{V}_j,t)+V_{ij}|\mca{V}|P(\widehat{\bm{n}}-\bm{V}_j,t)-V_{ij}|\mca{V}|P(\widehat{\bm{n}}-\bm{V}_j,t)-\widehat{n}_{i}P(\widehat{\bm{n}},t)\bigg)\\
	&=k_j(\widehat{n}_i-V_{ij})P(\widehat{\bm{n}}-\bm{V}_j,t)-k_j\widehat{n}_{i}P(\widehat{\bm{n}},t).
\end{align*}
When $f_{vj}(\bm{n},\omega)=k_j\omega$, we similarly have
\begin{align*}
	&\sum_{\bm{n}\in\mca{S}(\widehat{\bm{n}})}\sum_{v\in \mca{V}}\bigg(f_{vj}(\bm{n}-\widetilde{\bm{V}}_{vj},\omega)P(\bm{n}-\widetilde{\bm{V}}_{vj},t)-f_{vj}(\bm{n},\omega)P(\bm{n},t)\bigg)\\
	&=k_j\sum_{\bm{n}\in\mca{S}(\widehat{\bm{n}})}\sum_{v\in \mca{V}}\bigg(\omega P(\bm{n}-\widetilde{\bm{V}}_{vj},t)-\omega P(\bm{n},t)\bigg)\\
	&=k_j\bigg(\sum_{\bm{n}\in\mca{S}(\widehat{\bm{n}})}\sum_{v\in \mca{V}}\omega P(\bm{n}-\widetilde{\bm{V}}_{vj},t)-\sum_{\bm{n}\in\mca{S}(\widehat{\bm{n}})}\sum_{v\in \mca{V}}\omega P(\bm{n},t)\bigg)\\
	&=k_j\bigg(\sum_{v\in \mca{V}}\omega\sum_{\bm{n}\in \mca{S}(\widehat{\bm{n}})}P(\bm{n}-\widetilde{\bm{V}}_{vj},t)-\sum_{v\in \mca{V}}\omega\sum_{\bm{n}\in\mca{S}(\widehat{\bm{n}})}P(\bm{n},t)\bigg)\\
	&=k_j\Omega P(\widehat{\bm{n}}-\bm{V}_j,t)-k_j\Omega P(\widehat{\bm{n}},t).
\end{align*}
The master equation of $P(\widehat{\bm{n}},t)$ is then obtained as
\begin{equation}
\partial_t P(\widehat{\bm{n}},t)=\sum_{j=1}^{K}\left(f_j(\widehat{\bm{n}}-\bm{V}_j,\Omega)P(\widehat{\bm{n}}-\bm{V}_j,t)-f_j(\widehat{\bm{n}},\Omega)P(\widehat{\bm{n}},t)\right).
\end{equation}
\end{widetext}

\end{document}